\documentclass[runningheads]{llncs}
\usepackage{a4wide}
\usepackage{graphicx}
\usepackage{algorithm,algorithmicx}
\usepackage{amsfonts,amsmath,amssymb}

\begin{document}
\title{DXML: Distributed Extreme Multilabel Classification\thanks{Paper accepted in BDA 2021.}}
\titlerunning{DXML}
%
\author{Pawan Kumar\inst{1}\orcidID{0000-0001-5632-6964}}
\authorrunning{P. Kumar}
%
\institute{International Institute of Information Technology, Hyderabad, 500032, India \\
\email{pawan.kumar@iiit.ac.in}}
\maketitle              
\begin{abstract}
As a big data application, extreme multilabel classification has emerged as an important research topic with applications in ranking and recommendation of products and items. A scalable hybrid distributed and shared memory implementation of extreme classification for large scale ranking and recommendation is proposed. In particular, the implementation is a mix of message passing using MPI across nodes and using multithreading on the nodes using OpenMP. The expression for communication latency and communication volume is derived. Parallelism using work-span model is derived for shared memory architecture. This throws light on the expected scalability of similar extreme classification methods. Experiments show that the implementation is relatively faster to train and test on some large datasets. In some cases, model size is relatively small. \\
{\bf Code:} \url{https://github.com/misterpawan/DXML}

\keywords{extreme multilabel classification \and distributed memory \and multithreading.}
\end{abstract}

\section{Introduction}

Extreme multi-label learning is an active research problem in big data applications with several applications
in tagging, recommendation, and ranking. Consider a features matrix of $n$ examples $X \in \mathbb{R}^{n\times d}$ with $d$-dimensional features with their corresponding labels matrix $Y \in \mathbb{R}^{n\times \ell},$ the aim of multilabel classification problem is to assign some relevant labels out of a total of $\ell$-labels to new data sample. The extreme multilabel classification problem refers to the setting when $n,$ $d,$ and $\ell$ quickly scale to large numbers often upto several millions. Moreover, number of data samples also are in several millions. This is one of the classic challenge problems of big data analytics. 

There are several challenges in designing algorithms for extreme multilabel classification. It is found that the average number of labels per data points is usually very small, moreover, it is know that label frequencies follow the so-called Zipf's law, which means that there are only small number of labels which are found in large number of samples, such labels are called head labels. On the other hand, there are large number of lables that occur less frequently, and such labels are called tail labels. Such a distribution creates a bias in the classifier, because since head labels occur more frequently, the classifier may learn more robustly about predicting head labels compared to tail labels, thereby, leading to a classifier that is biased towards predicting head labels better. Despite all these challenges, one of the major concerns is designing scalable algorithms for modern day hardwares.      

Looking at the increasing trend of number of labels going upto millions, a hybrid parallel design and implementation of extreme classification algorithms is essential that can exploit shared as well as distributed memory architectures \cite{kumar2014,kumar2015,kumar2013,kumar2013b,kumar2014b}. In this paper, we show a parallel design and implementation for a hybrid distributed memory and shared memory implementation using MPI and OpenMP. To the best of our knowledge, this is the first time a hybrid parallel implementation has been shown for extreme classification. We derive the communication bounds for distributed memory, and parallelism bound for shared memory implementation. Our preliminary numerical experiments suggest that we have fastest training and test times when compared to some of the existing parallel implementations in C/C++. 

The rest of the paper is organized as folows. In section 2, we discuss previous related work. In section 3, we discuss distributed memory implementation. We derive expressions for communication volume and latency. Finally, in section 4, we discuss numerical experiments on multilabel datasets. 

\section{Previous Work}
Some papers for multilabel classification were proposed during 2006-2014 \cite{Tsoumakas2007,Zhang2014}. Some of these papers explored $k$ nearest neighbours \cite{Zhang2007}. The idea of using random forest were also proposed \cite{Kocev2007}. With the recent demand for scaling the algorithms to millions of labels, new class of methods have been proposed, where scalability is achieved by either parallelism \cite{Yen2016,Yen2017},  dimension reduction \cite{Bhatia2015,Tagami2017,Weston2011,Yu2014}, or by hierarchical embedding \cite{Jain2016,Jasinska2016,Prabhu2014,Weston2013,naram2021a}.   

\section{Distributed Memory Implementation}
We show a hybrid parallel (MPI+OpenMP) implementation of CRAFTML \cite{craftml}, and call it DXML. During training, DXML computes a forest $F$ of $m_F$ $k$-ary instance trees, which are constructed by recursive partitioning. The training algorithm is shown in Algorithm \ref{alg:train}. In line 1, the input is a feature matrix $X$ and a label matrix $Y.$ We wish to build a label-tree with nodes denoted by $v.$ We then apply the termination condition. The termination condition of the recursive partitioning are the following
\begin{enumerate}
    \item cardinality of the node's instance subset is less than a given threshold $n_{\text{leaf}}$
    \item all the given instances have the same features
    \item all the given instances have the same labels
\end{enumerate}

If the stop condition is false, and the current node $v$ is not a leaf as in line 4, then a multi-class classifier is built using Algorithm \ref{alg:classifier}. 

The sequential node training stage in DXML will be decomposed into following three consecutive steps:
\begin{itemize}
\item a random projection into lower dimensional spaces of the features and label vectors corresponding to the node's instances. That is, in Algorithm \ref{alg:classifier}, in line 2, we first sample $X_s$ and $Y_s$ from $X_v$ and $Y_v$ with a sample size $n_s.$ Then in lines 3 and 4, we do the random projection using projection feature matrix $P_x,$ and a random label projection matrix $P_y.$  
    \item from projected labels, partitioning of of the corresponding instances into $k$ temporary subsets using $k-$means. 
    \item A multiclass classifier is trained to assign each instance to the relevant temporary subset (that is, the cluster index found at step 2 above) from the corresponding feature vector. The classifier then partitiones the instances into $k$ final subsets or child nodes. This is achieved in by first calling a splitting (line 6, Algorithm \ref{alg:train}) of the instances into child nodes using output of k-means of Algorithm \ref{alg:classifier}. Then trainTree function is called recursively on the child nodes.  
\end{itemize}

Similar to FastXML, the nodes partitioning objective of DXML is to regroup instances with common labels in a same subset, but the computation is different. Finally, once a tree has been trained, its leaves store the average label vector of their associated instances. The partition strategy is driven by two constraints: partition computation must be based on randomly projected instances to ensure diversity, and must perform low complexity operations for scalability. The training algorithm is given below.
\begin{algorithm}
\caption{\label{alg:train}trainTree}
\begin{algorithmic}[1]
    \State \textbf{Input:} Training set with a feature matrix $X$ and a label matrix $Y.$
    \State Initialize node $v.$
    \State $v.$isLeaf $\leftarrow$ testStopCondition($X,Y$)
    \State if $v.$isLeaf = false then
    \State \quad $v.$classif $\leftarrow$ trainNodeClassifier($X,Y$)
    \State \quad $(X_{child_i},Y_{child_i})_{i=0,\cdots,k-1}$ $\leftarrow$ split($v.$classif,$X,Y$)
    \State \quad for $i$ from $0$ to $k-1$ do
    \State \quad \quad $v.child_i \leftarrow$ trainTree($X_{child_i},Y_{child_i}$)
    \State \quad end for
    \State else
    \State \quad $v.\hat{y} \leftarrow $ computeMeanLabelVector($Y$)
    \State end if
    \State \textbf{Output:} node $v$
\end{algorithmic}
\end{algorithm}

\begin{algorithm}
\caption{\label{alg:classifier}trainNodeClassifier}
\begin{algorithmic}[1]
    \State \textbf{Input:} feature matrix $(X_v)$ and label matrix $(Y_v)$ of the instance set of the node $v.$
    \State $X_s,Y_s \leftarrow$ sampleRows($X_v,Y_v,n_s$)   \Comment{$n_s$ is the sample size}
    \State $X_{s}^{'} \leftarrow X_sP_x$   \Comment{random feature projection}
    \State $Y_{s}^{'} \leftarrow Y_sP_y$  \Comment{random label projection}
    \State $c \leftarrow k$-means($Y_{s}^{'},k$)  \Comment{$c \in \{0,\cdots,k-1\}^{\min(n_v,n_s)}$} \\
    \Comment{$c$ is a vector where the $j^{th}$ component $c_j$ denotes the cluster idx of the $j^{th}$ instance associated to $(X_{s}^{'})_{j,.}$ and $(Y_{s}^{'})_{j,.}.$}
    \State for $i$ from $0$ to $k-1$ do
    \State \quad $(classif)_{i,.} \leftarrow $computeCentroid(${(X_{s}^{'})_{j,.}| c_j = i}$)
    \State end for
    \State \textbf{Output:} Classifier $classif(\in \mathbb{R}^{k \times d_{x}^{'}})$
\end{algorithmic}
\end{algorithm}

\subsection{Some More Detail on Training}
\textbf{Step 1: Random projections of the instances of $v$:} The feature and the label vectors $x$ and $y$ of each instance in $v$ are projected into a space with a lower dimensionality: $x^{'} = XP_x$ and $y^{'} = yP_y$ where $P_x$ and $P_y$ are random projection matrices of $\mathbb{R}^{d_x \times d_{x}^{'}}$ and $\mathbb{R}^{d_y \times d_{y}^{'}}$ respectively, and $d_{x}^{'}$ and $d_{y}^{'}$ are the dimensions of the reduced feature and label spaces resp. The projection matrices are kept different from one tree to another. The random projection considered is a sparse orthogonal projection matrix \cite{hashing_weinberger} with values of $-1$ or $+1$ on each row. The sparsity of the projections lead to faster computations, hence, faster projections. To retain the sparsity we use the hashing; we describe it next. 

In this so-called hashing trick algorithm, high dimensional dataset is projected into a lower dimension. Projection is done row after row in this algorithm. In a row of original matrix each element index is considered as a key and each element index of corresponding row in the projected matrix is considered as the bucket. Each index of original row (key) is mapped to index of lower dimension projected row (bucket) using the hash function. Similarly each key is also mapped to one of the two signs ($ + $ or $ - $) using another hash function. Multiple keys may be mapped to the same bucket, in that case elements in the same bucket are multiplied by their respective signs (obtained from second hash function) and added. Below is the algorithm of Hashing Trick.

\begin{algorithm}[H]
\caption{\label{alg:hash}Hashing Trick}
\begin{algorithmic}[1]
    \State \textbf{Input:} $X_{s}$, projectedSpaceDimension, $S_{x}$, $SS_{x}$
    \State $X_{s}' $ $\leftarrow$ 0 \Comment{initialize projected data matrix}
    \State $p$ $\leftarrow$ projectedSpaceDimension
    \State $nR$ $\leftarrow$ $X_{s}$.numberOfRows()
    \State $nC$ $\leftarrow$ $X_{s}$.numberOfCols()
    \State for $i$ $\leftarrow$ $0$ to $nR:$
    \State \quad for key $\leftarrow$ 0 to $nC:$
    \State \quad \quad Index $\leftarrow$ $hash1(key,S_{x})\%p$
    \State \quad \quad Sign $\leftarrow$ $2 \times (hash2(key,SS_{x})\%2) - 1$
    \State \quad \quad $X_{s}'$[i,Index] $\leftarrow$ $X_{s}'$[i,Index] + (Sign* $X_{s}$[i,key])
    \State \textbf{Output:} $X_{s}'$ 
\end{algorithmic}
\end{algorithm}

This Algorithm \ref{alg:hash} takes $X_{s}$ ( or $ Y_{s} $ ) and the projectedSpaceDimension, and  $S_{x},SS_{x}$ (seeds for hash1, hash2 functions respectively), and returns projected samples $X_{s}'$ (or $ Y_{s}' $). In line 4, $nR$ is the number of rows in $ X_{s},$ and $nC$ is the number of columns in $ X_{s}.$ The lines 6 and 7 are loops with $i$ iterating over the row indices of $X_{s}$, and \textbf{key} iterating over the column indices in the row (with index $i$).
In line 8, the Column Index (key) is hashed using hash1 function and if its more than the projectedSpaceDimension, then remainder when divided by p is considered to map it in the range$(0,p).$ In line 9, $2 \times hash2(key)\%2 - 1$ maps key to $+1$ or $-1$, and finally in
line 10, the element $X_{s}[i,key]$ is multiplied by the sign, and then added to already present element at $X_{s}'[i,Index].$

We consider the case where feature and label projections are the same in each node of $T.$
We may have tried different projections per node, but we don't consider that in this paper. 

\textbf{Step 2: Partitioning of Instance into $k$ Temporary Subsets :} Let $Y_s$ be the label matrix of a sample drawn without replacement. Let the sample size be at most $n_s.$ This sample is partitioned with a spherical $k$-means applied on $Y_sP_y.$ The use of spherical $k$-means is motivated by the facts that it is well-adapted to sparse data, moreover, the cosine metric is fast to compute. 

The cluster centroids are initialized using the $k$-means++ strategy to ameliorate the cluster stability, and to improve the algorithm performance against a random initialization.

The $k$means++ initialization strategy is as follows:
\begin{itemize}
    \item[1. ] Among the label centers, choose one center uniformly at random. 
    \item[1.] Choose one center uniformly at random from among the label vectors.
    \item[2.] For each label vector $y$, compute the distance $D(y),$ defined to be the distance between $y$ and the closest center that has already been chosen above.
    \item[3.]Choose one new data point at random as a new center, using a weighted probability distribution where a point $y$ is chosen with probability proportional to $D(y)^2$.
    \item[4.] Repeat steps 2 and 3 above until all centers have been chosen.
    \item[5.] After all the centers have been chosen, proceed with the spherical $k$-means clustering.
\end{itemize}

\textbf{Step 3: Assigning a subset from the projected features:} In each temporary subset the centroid of the projected feature vectors is computed. During the prediction phase, 
if the centroid of the subset is closest to the projected feature vector, then the classifier assigns this subset. For computing the closeness, the cosine measure is used.

\textbf{Prediction:} In the prediction phase, for each tree, the input sample goes from root to leaf, which is determined by the successive decisions of the classifier. The prediction is the average label vector stored in the leaf reached. The forest then aggregates the tree predictions with the average operator.

\paragraph{Algorithm Analysis}
Let $s_x(s_y)$ denote the average number of non-zero elements in the feature (label) vectors of the instances. Due to the hashing trick, the projected feature and label vectors have less than $s_x$ and $s_y$ non-zero elements in average. For a node $v$ of a tree $T,$ let $n_v$ denote the number of instances of the subset associated to $v.$ Let $i$ be the number of iterations of the spherical $k$-means algorithm.

\begin{lemma}
For a node $v$ of a tree $T,$ the time complexity $C_v$ is bounded by $O(n_v \times C)$ where $$C = k \times (i \times s_y + s_x)$$ is the complexity per instance.
\end{lemma}
\begin{proof}
See \cite{craftml}.
\end{proof}

Let $T$ be a strictly $k$-ary tree and $\ell_T$ be its number of leaves. Let $m_T = \dfrac{\ell_T -1}{k-1}$ be its number of nodes and $\Bar{n}_T = \dfrac{\sum_{v \in T}n_v}{m_T}$ be the average number of instances in its nodes.

\begin{proposition}
\label{trainComplexity}
If the tree $T$ is balanced, its training time complexity $C_T$ is bounded by $$O\left(\log_{k}\big(\frac{n}{n_{leaf}}\big)\times n \times C\right).$$ 
Otherwise, $C_T$ is equal to $$O \left(\frac{\ell_T -1}{k -1} \times \Bar{n}_T \times C \right).$$ 
\end{proposition}
\begin{proof}
See \cite{craftml}.
\end{proof}

It can be observed that the time complexities are independent of the projection dimensions $d_{x}^{'}$ and $d_{y}^{'}.$ The clustering is done after sampling from the instances, thus the training and clustering complexity are further reduced.

\begin{proposition}
The memory complexity of a tree $T$ is bounded by $$O \big(n \times s_y + m_T \times k \times d_{x}^{'} \big).$$
\end{proposition}
\begin{proof}
See \cite{craftml}.
\end{proof}

\subsection{Hybrid MPI and OpenMP Parallel Implementation}
We use message passing interface MPI \cite{OpenMPI} to train for each learner. Each learner or process reads the data, and calls trainTree in Algorithm \ref{alg:train}. The classification for each node (child) at a given level is processed by multiple threads. Each learner stores their own mean label vector computed in line 11 of Algorithm \ref{alg:train}. The model parameters for each tree is sent to master node for faster prediction. Let $n_t$ be the model size for each tree, and there are $m_F$ trees, then the communication volume from the worker nodes to the master node is $O(m_F n_t).$ The latency cost is the number of messages passed. In this case, each processor communicates once to master node. Hence the latency cost is $O(m_F).$ Let $P$ denote the number of processors, then since there are as many processors as number of trees in the forest, $P = m_F.$ We have the following results. 
\begin{proposition}
Let there be $P$ processors and $m_F = P,$ we have the following communication costs
\begin{align*}
\text{communication\_volume}~&=O(P n_t) \\
\text{latency}~&= O(P)
\end{align*}
\end{proposition} 
We also exploit shared memory parallelism using OpenMP \cite{OpenMP} when training each tree. This follows a spawn-sync model. At root, a master thread launches $k$ child processes, and each of the $k$ child process calls trainNodeClassifier. We use work-span model \cite{cilk} to do parallel complexity 
analysis of trainTree. The span denoted by 
$T_{\infty},$ which corresponds to critical path, i.e., the longest path from root node to a leaf node. Now we calculate the work denoted by $T_1,$ which is the total work done to train the classifier for all the nodes. We define the parallelism to be $T_1/T_{\infty}.$ We have the following proposition. 
\begin{proposition}
Let $T$ be a strictly $k-$ary tree, $\ell_T$ be its number of leaves. Let $m_T, \bar{n}_T,$ and $C$ be defined as above, then the parallelism for a tree is bounded by 
\begin{align*}
T_1 / T_{\infty} = O \left( \dfrac{1}{\log_k m_T} \log_k \left( \dfrac{n}{n_{leaf}} \right) \times n \times C \right). 
\end{align*}
\end{proposition}  
\begin{proof}
From \ref{trainComplexity}, the total work done which is denoted by $C$ is given by the following bound 
\begin{align*}
T_1 = C = O \left( \log_k \left( \dfrac{n}{n_{leaf}} \right)\times n \times C \right). 
\end{align*}
We also have $T_{\infty} = \log_k m_T.$ This gives the required parallelism.
\end{proof}

\section{Numerical Experiments}
We did our MPI+OpenMP experiments on Intel Xeon architecture with 10 nodes with 120GB RAM. We used 10 MPI processes and 5 threads per processes. We choose $m_F=50$ and $n_{leaf}=10.$ The feature and label projection dimensions are $d_x' = \min(d_y, 10000)$ and $d_y' = \min(d_y, 10000).$
We show the precision scores in Table \ref{tab:precision}. The best precision scores among the tree-based classifiers are indicated in bold. For example, the DXML has highest P@1 precision scores for Mediamill, EURLex-4K, Delicious-200K among the tree based classifiers. In other cases, it is close to the precisions of other classifiers. In Table \ref{tab:time}, we show the train time, test time, and model size. The train times for DXML was best among all methods for all the datasets.  The model size for DXML is same as the one for CRAFTML. The model size for DXML/CRAFTML was best for Amazon-670 and Delicious-200K. The model size for PPDSp is very large for Amazon-670. For DXML, the learned model parameters remain distributed, hence, there is an additional cost of reduction operation. 

\begin{table*}[h]
    \centering
   
    \begin{tabular}{cc|c|ccc|c}
       \hline
      \multicolumn{2}{c|}{Language} &  \multicolumn{4}{c|}{C/C++} & \multicolumn{1}{c}{C++} \\ \hline 
     \multicolumn{2}{c|}{Machine}   & \multicolumn{1}{c|}{ 50 cores} & \multicolumn{3}{c|}{1 core} &  \multicolumn{1}{c}{100 cores}  \\  
     \hline 
    Algorithm &                                & DXML      & FastXML   & PFastReXML    & SLEEC       & DISMEC     \\
    \hline 
      EURLex-4K         & Train                 & {\bf 38.01}     & 315.9         & 324.4     & 4543.4     & 76.07      \\
                        & Test (ms)            & {\bf 1.29}      & 3.65          & 5.43      & 3.67        & 2.26       \\
                        & Model (MB)         & 30        & 384           & NA        & 121           & {\bf 15}         \\ \hline 

      Delicious-200K    & Train                 &  {\bf 2929.0}         & 8832.46       & 8807.51   & 4838.7    & 38814      \\
                        & Test (ms)              & 10.40          & 1.28          & 7.4       & 2.685      & 311.4    \\
                        & Model (GB)       &    {\bf 0.346}       & 1.3           & 20        & 2.1       & 18        \\ \hline 
      Amazon-670K       & Train                 & {\bf 752.65}    & 5624          & 6559      & 20904           & 174135     \\
                        & Test (ms)            & 3.65      & {\bf 1.41}          & 1.98      & 6.94          & 148      \\
                        & Model (GB)       & {\bf 0.494}      & 4.0           & 6.3       & 6.6            & 8.1       \\ \hline 
      AmazonCat-13K     & Train                & {\bf 1164.13}   & 11535         & 13985     & 119840      & 11828     \\
                        & Test (ms)           & 8.98      & 1.21          & 1.34      & 13.36       & 0.2        \\
                        & Model (GB)      & {\bf 0.659}     &9.7            & 11        &12      & 2.1       \\
      \hline 
    \end{tabular}
    \caption{Train Time, Test Time, and Model Size. Here NA means not available.}
    \label{tab:time}
    
\end{table*}

\begin{table*}[h]
    \centering 
    \begin{tabular}{cc|cccc|cccc}
   \hline
      \multicolumn{2}{c|}{Method Type} & \multicolumn{4}{c|}{Tree based} & \multicolumn{2}{c}{Other} \\ \hline 
      Algorithm     &   Scores   & DXML    & PFastReXML& FastXML & LPSR     & SLEEC         & DISMEC\\
    \hline 
       Mediamill    &  P@1 & {\bf 87.20}   & 83.98   & 84.22     & 83.57    & 87.82         & 84.83 \\
                    &  P@3 & {\bf 71.52}     & 67.37   & 67.33     & 65.78    & 73.45         & 67.17 \\
                    &  P@5 &   {\bf 57.56}    & 53.02   & 53.04     & 49.97    & 59.17        & 52.80 \\ \hline 

       Delicious    &  P@1 &  67.28   & 67.13   & {\bf 69.61}     & 65.01    & 67.59            & NA \\
                    &  P@3 &61.64     & 62.33   & 64.12     & 58.96    & 61.38             & NA \\
                    &  P@5 & 57.17    & 58.62   & 59.27     & 53.49    & 56.56             & NA \\ \hline 
       EURLex-4K    &  P@1 &  {\bf 78.20}   & 75.45   & 71.36     & 76.37    & 79.26         & 82.40 \\
                    &  P@3 &  64.2     & 62.70   & 59.90     & 63.36    & 64.30         & 68.50 \\
                    &  P@5 &  53.26    & 52.51   & 50.39     & 52.03    & 52.33         & 57.70 \\ \hline 
       Wiki-10      &  P@1 & {\bf 84.68}    & 83.57   & 83.03     & 72.72    & 85.88           & 85.20 \\
                    &  P@3 &  72.54    & 68.61   & 67.47     & 58.51    & 72.98             & 74.60 \\
                    &  P@5 &  62.47    & 59.10   & 57.76     & 49.50    & 62.70             & 65.90 \\ \hline 
 
      Delicious-200K&  P@1 &   {\bf 47.8}6  & 41.72   & 43.07     & 18.59    & 47.85        & 45.50 \\
                    &  P@3 &  41.25     & 37.83   & 38.66     & 15.43    & 42.21         & 38.70 \\
                    &  P@5 &  38.00    & 35.58   & 36.19     & 14.07    & 39.43         & 35.50 \\ \hline 
      Amazon-670K   &  P@1 &  37.31   & {\bf 39.46}   & 36.99     & 28.65    & 35.05        & 44.70  \\
                    &  P@3 &  33.30   & 35.81   & 33.28     & 24.88    & 31.25         & 39.70 \\
                    &  P@5 &    30.50 & 33.05   & 30.53     & 22.37    & 28.56        & 36.10 \\ \hline 
      AmazonCat-13K &  P@1 & 92.69     & 91.75   & {\bf 93.11}     & NA        & 90.53        & 93.40 \\
                    &  P@3 & 78.43     & 77.97   & 78.20     & NA        & 76.33         & 79.10  \\
                    &  P@5 &   63.56 & 63.68   & 63.41     & NA       & 61.52         & 64.10 \\
      \hline 
    \end{tabular}
    \caption{Precision Scores. Here NA stands for not available.}
    \label{tab:precision}
    
\end{table*}

\section{Conclusion}
For large scale recommendation problem using extreme multilabel classification, a scalable recommender model is essential. We proposed a hybrid parallel implementation of extreme classification using MPI and OpenMP that can scale to arbitrary number of processors. The best part is that this does not involve any loss function or iterative gradient methods. Our preliminary results show that our training time is fastest for some datasets. We derived the communication complexity analysis bounds for both the shared and distributed memory implementations.  With more cores, our parallel analysis suggests that the proposed implementation has the potential to scale further.

\section{Acknowledgement}
This work was done at IIIT, Hyderabad using IIIT seed grant. The author acknowledges all the support by institute. 

\bibliographystyle{splncs04}
\bibliography{ref}

\end{document}